\documentclass[11pt]{article} % DO NOT CHANGE THIS
\usepackage{times}  % DO NOT CHANGE THIS
\usepackage{helvet}  % DO NOT CHANGE THIS
\usepackage{courier}  % DO NOT CHANGE THIS
\usepackage[hyphens]{url}  % DO NOT CHANGE THIS
\usepackage{multicol}
\usepackage{graphicx} % DO NOT CHANGE THIS
\usepackage{fullpage}
\usepackage{caption} % DO NOT CHANGE THIS AND DO NOT ADD ANY OPTIONS TO IT
\DeclareCaptionStyle{ruled}{labelfont=normalfont,labelsep=colon,strut=off} % DO NOT CHANGE THIS
\usepackage{algorithm}
\usepackage{algorithmic}

\usepackage{newfloat}
\usepackage{listings}
\floatstyle{ruled}
\newfloat{listing}{tb}{lst}{}
\floatname{listing}{Listing}

\usepackage{amssymb}
\usepackage{amsmath}
\usepackage{fullpage}
\usepackage{caption}
\usepackage{wrapfig}
\usepackage[belowskip=-8pt]{caption}
\usepackage{wrapfig}
\usepackage{subcaption}
\usepackage[square,numbers,sort]{natbib}
\usepackage{color}
\usepackage{authblk}
\usepackage{colortbl}
\definecolor{webgreen}{rgb}{0,0.4,0}
\definecolor{webbrown}{rgb}{0.6,0,0}
\definecolor{purple}{rgb}{0.5,0,0.25}
\definecolor{darkblue}{rgb}{0,0,0.7}
\definecolor{darkred}{rgb}{0.7,0,0}
\usepackage{mathtools}
\usepackage{hyperref}
\hypersetup{colorlinks,citecolor=darkblue,filecolor=black,linkcolor=darkred,urlcolor=webgreen, pdfstartview={XYZ null null 1.00}}
\newcommand{\ignore}[1]{}

\newtheorem{lemma}{{Lemma}}
\newtheorem{proposition}{{\sc Proposition}}
\newtheorem{corollary}{{Corollary}}
\newtheorem{theorem}{{Theorem}}
\newtheorem{definition}{{Definition}}

\newtheorem{fact}{{Fact}}

\usepackage{cleveref}
\crefname{claim}{claim}{claims}
\crefname{fact}{fact}{facts}
\crefname{algorithm}{algorithm}{algorithms}
\crefname{observation}{observation}{observations}
\crefname{equation}{equation}{equations}
\crefname{assumption}{assumption}{assumptions}
\crefname{lemma}{lemma}{lemmata}
\crefname{corollary}{corollary}{corollaries}
\crefname{figure}{fig.}{figs.}

\newenvironment{proof}{\noindent {\em Proof\/}:\enspace}
{\hfill $\blacksquare{}$ \medskip \\}

\DeclareMathOperator*{\argmax}{\arg\!\max}

\usepackage{pgf}
\usepackage{verbatim}
\usepackage{enumerate}
\usepackage{amssymb}
\usepackage{mathrsfs}
\usepackage{algorithm}
\usepackage{resizegather}
\newif\ifverbose
\verbosetrue %%% Comment / Uncomment this for viewing without / with author comments
\newcommand{\sn}[1]  {\ifverbose {\noindent \textcolor{blue}{{\bf SN: }{\em #1}} } \else \fi }

\usepackage{url}

\usepackage{xspace}

\renewcommand{\AA}{\ensuremath{\mathcal A}\xspace}

\newcommand{\DD}{\ensuremath{\mathcal D}\xspace}

\newcommand{\GG}{\ensuremath{\mathcal G}\xspace}

\newcommand{\MM}{\ensuremath{\mathcal M}\xspace}
\newcommand{\NN}{\ensuremath{\mathcal N}\xspace}

\newcommand{\PP}{\ensuremath{\mathcal P}\xspace}
\newcommand{\QQ}{\ensuremath{\mathcal Q}\xspace}

\renewcommand{\SS}{\ensuremath{\mathcal S}\xspace}

\newcommand{\aaa}{\ensuremath{\mathfrak a}\xspace}

\newcommand{\sss}{\ensuremath{\mathfrak s}\xspace}

\usepackage{nicefrac}

\newcommand{\mech}{\text{\bf \texttt{VCG-T}}}
\newcommand{\mechfull}{{\bf \texttt{VCG}} with {\bf \texttt{T}}ime delays}

\newcommand{\multi}{\text{\bf \texttt{MIA}}}
\newcommand{\maafull}{\text{\bf \texttt{MIA}} {\bf \texttt{A}}pproximation {\bf \texttt{A}}lgorithm}
\newcommand{\MUCA}{\text{\bf \texttt{MUCA}}}
\newcommand{\MAA}{\text{\bf \texttt{MAA}}}
\newcommand{\hi}{\text{\bf \texttt{high}}}
\newcommand{\med}{\text{\bf \texttt{medium}}}
\newcommand{\lo}{\text{\bf \texttt{low}}}

\title{Social Distancing via Social Scheduling}

\author[1]{Deepesh Kumar Lall}
\author[1]{Garima Shakya}
\author[1]{Swaprava Nath}

\affil[1]{\small Indian Institute of Technology Kanpur, \texttt{\{deepeshk, garimas, swaprava\}@iitk.ac.in}}

\begin{document}

\maketitle

\begin{abstract}
Motivated by the need of {\em social distancing} during a pandemic, we consider an approach to schedule the visitors of a facility (e.g., a general store). Our algorithms take input from the citizens and schedule the store's discrete time-slots based on their importance to visit the facility. Naturally, the formulation applies to several similar problems. We consider {\em indivisible} job requests that take single or multiple slots to complete. The salient properties of our approach are: it (a)~ensures social distancing by ensuring a maximum population in a given time-slot at the facility, (b)~aims to prioritize individuals based on the importance of the jobs, (c)~maintains truthfulness of the reported importance by adding a {\em cooling-off} period after their allocated time-slot, during which the individual cannot re-access the same facility, (d)~guarantees voluntary participation of the citizens, and yet (e)~is computationally tractable. The mechanisms we propose are prior-free. We show that the problem becomes NP-complete for indivisible multi-slot demands, and provide a polynomial-time mechanism that is truthful, individually rational, and approximately optimal. Experiments with data collected from a store show that visitors with more important (single-slot) jobs are allocated more preferred slots, which comes at the cost of a longer cooling-off period and significantly reduces social congestion. For the multi-slot jobs, our mechanism yields reasonable approximation while reducing the computation time significantly.
\end{abstract}
\section{Introduction}
\label{sec:intro}

Pandemics like the COVID-19 showed that one of the most effective solutions against infectious diseases is {\em social distancing}~\cite{wilder2020isolation}. Therefore, it is a practice we ought to master and stay prepared as a preventive disease containment measure in the future. While most citizens are willing to follow social distancing, the lack of communication and coordination among them, particularly prior to an immediate lockdown, overcrowds shopping centers even though the footfalls could have been evenly distributed over the shop's working hours to maintain a medically recommended social distance. 

In our setting, each customer has an infinite queue of jobs that have different importances (privately known only to the customer) and lengths. However, the customers are {\em myopic}, i.e., worries only about the last unprocessed job. They experience a better value if the job is assigned their preferred slots, but also a {\em disutility} to wait before submitting their next job for allocation. All jobs are {\em indivisible}, i.e., has to be completed once started.
In this paper, we consider {\em two} settings: (i)~all jobs are of single time-slot length, (ii)~different jobs are of different integral time-slot lengths.

Though cast in the context of social scheduling for pandemics, a similar problem arises in general scheduling settings, e.g., scheduling traffic in freeways~\cite{aydos2014scats}
or multi-ownership computational jobs in a single-core processor~\cite{Sahni76}.
Since all such settings have multiple agents competing for a common resource and the importance of the jobs are private, the solutions involving truthful revelation in a computationally tractable manner also apply to those settings.

This paper introduces a novel approach to pandemic containment using mechanism design that reduces the congestion in facilities, satisfies various desirable theoretical properties, and exhibits fair performance in practice. The following section details the contributions of this paper.

\subsection{Brief Problem Description and Contributions}
\label{sec:contributions}

The opening hours of a facility are divided into {\em periods} (e.g., a day), each of which has multiple {\em slots} (e.g., every hour when it is open). The customers have an unlimited number of outstanding jobs to be processed in a sequence at the facility, and they report the valuations of the immediate unprocessed job.\footnote{A typical shopper knows that she needs to visit a store many times but precisely knows the importance of the immediate visit.} A valuation $v_{ij}$ denotes agent $i$'s importance for that job if it starts in slot $j$. Since this information is private to agent $i$, a mechanism needs to elicit this truthfully. 
In a setting where the agents' preferences are private, if the mechanism has no additional structures (e.g., transfers of payoff), only {\em dictatorial} mechanisms (where a pre-selected agent's favorite outcome is always selected) are truthful~\cite[Thm 7.2]{Roberts79}. 
Therefore, the use of transfers in some form is inevitable to ensure truthfulness of the agents. However, for pandemic containment, the use of money for scheduling citizens is unethical and illegal in certain countries. Hence, we use {\em time-delay} as a replacement of money. Waiting time is often seen as a resource that individuals agree to trade with~\cite{leclerc1995waiting}. Our scheduling approach will work in all places where payment can be replaced with a time-delay. 
Quite naturally, an agent prefers to have a more valuable slot assigned to her with less time-delay. We model the agents' payoffs using the well-known quasi-linear payoff model \cite[Chap 10]{SL08}.
This competitive scenario induces a {\em game}
where agents' actions are to report the valuations. The agents may {\em overstate} (or {\em understate}) their actual valuations.
The contributions of this paper can be summarized into the following {\em four} major points:
\begin{itemize}
    \item Even though maximizing the {\em social welfare} (sum of the agents' valuations) of the slot allocation is a combinatorial optimization problem; we show that this problem is computationally easy to solve for jobs with single-slot length (\Cref{thm:relax,thm:strong}). 
    \item We show that the standard Vickrey-Clarke-Groves (VCG) payment~\citep{Vick61,Clar71,Grov73} can be used as the delay (cooling-off time) in this setting to ensure truthfulness (\Cref{thm:dsic}) and participation of the agents (\Cref{thm:ir}). We call the allocation and delay together as the mechanism \mech\ (\mechfull).
    \item The welfare maximizing allocation of multi-slot jobs are computationally hard (\Cref{thm:NPc}). We propose a polynomial time mechanism (\Cref{lem:alg2complexity}) which ensures participation (\Cref{lem:alg2IR}), truthfulness (\Cref{lem:alg2IC}), and is approximately optimal (\Cref{thm:apprx_multi}).
    \item Our real and synthetic data experiments (\Cref{sec:experiment}) show that visitors with more important jobs are allocated more preferred slots, which comes at the cost of a longer delay to re-access the store. We show that social distancing is significantly improved using users' visit data from a store (\Cref{sec:reduction}). For the multi-slot jobs, our approximately optimal mechanism provides a reasonable approximation at a much reduced computational cost in practice (\Cref{sec:approx-plot}).
\end{itemize}
The mechanisms presented in this paper are {\em prior-free}, i.e., they do not depend on the probabilistic information of the valuations.

\subsection{Related Work}
\label{sec:literature}
% \sn{needs to be reduced by one column}
Social distancing measures have been widely successful and recommended for pandemic control~\cite{glass2006targeted,thunstrom2020benefits,fong2020nonpharmaceutical}.
However, it is also observed that the benefits of social distancing depend on the extent to which the citizens follow it. An extensive part of the recent research related to social distancing during COVID19 aims to understand the relationship of social distancing with different public policies and other factors~\cite{gosak2021endogenous,arazi2021discontinuous,cho2020mean,pejo2020corona}. 
%\cite{pejo2020corona} model the coronavirus situation as a game between rational agents where the strategies are whether to wear a mask or not, to go out for a meeting or not, and others. They show how the decision of an agent to go out to attend a meeting at the risk of her own life depends on the maximum duration, the maximum number of people allowed in a gathering, and the probability of getting infected.
Individuals are sometimes reluctant to pay the costs inherent in social distancing which can limit its effectiveness as a control measure~\cite{reluga2010game}. \citet{toxvaerd2020equilibrium} considers social distancing in the susceptible-infected-recovered (SIR) epidemiological model and 
%characterizes the equilibrium dynamics in terms of peak prevalence, time-domain properties of the disease, and other factors. \citet{toxvaerd2020equilibrium}
shows that ``during the equilibrium social distancing phase, individuals gradually reduce their social distancing efforts despite the infection probability not decreasing". \citet{cavallo2021social} shows that under an uncoordinated model, every equilibrium involves more social contact than it occurs in a social optima.
A related thread of social distancing measures exists in workforce-intensive organizations in the forms of rostering, workplace design, or allowing people to work from home~\cite{roycroft2020rostering}. However, these measures are centrally controlled, and individuals do not have a scope to behave strategically. On the other hand, the businesses maintain social distancing by enforcing certain appointment booking apps (web or smartphone-based) on customers. These apps run a {\em first-come-first-served} algorithm for slot booking and do not provide any priority preservation guarantees.\footnote{An example app is \url{www.appointy.com}.}

% Resource allocation is a related problem where a lot of game-theoretic ideas have been used in the literature. For instance, cloud computing~\cite{nezarat2015efficient}, grid computing~\cite{khan2006non}, wireless networks~\cite{zhou2014auction}, and many more application domains use ideas from game theory and auction theory.
%
% \textbf{A related line of work is to allocate independent tasks to unrelated machines to minimize the makespan, which is the time needed to finish all the tasks \cite{christodoulou2009mechanism}. Nisan and Ronen~\cite{NISAN2001166} gave an n-approximation deterministic and truthful mechanism. For a restricted case, when all tasks have different low and high values,  Lavi and Swamy~\cite{lavi2009truthful} devised a 3-approximation truthful in expectation mechanism. And for identical low or high values, they gave a 2-approximation deterministic cycle-monotone algorithm. The minimum makespan problem is different from scheduling at a facility with restricted capacity. One way to look at the similarities and differences is to consider each unit capacity of the facility as a single machine. In the minimum makespan problem, we observe the valuation matrix is for each task on each machine. In scheduling at a facility, the agents do not care which part of the capacity is assigned to them, but at what time of the opening hours they can visit the facility. The valuation matrix has different dimensions in both the problems and the solutions are for a different objective than that in our setting.}

 Coordination among the citizens is essential for social distancing \cite{shadmehr2020coordination}, but the technology is still not much developed to address the social coordination problem, particularly when the participants are self-interested independent decision makers. Mechanism design is an approach which can equip an artificially intelligent app to satisfy certain desirable properties on the face of the participants' strategic behavior. This is attempted in the current paper.

Resource allocation with monetary transfers to ensure truthfulness, e.g., in the context of jobshop scheduling \cite{hajiaghayi2005online}, has a rich literature that is close to our work. 
%  \cite{friedman2003pricing} consider interruptible and non-interruptible versions of the online {\em single} resource allocation problem, where the resources can be allocated to a limited number of users, e.g., the WiFi at Starbucks or shared computation on a server. They assume that the agents can not lie about their arrival time to get better valuations and prove that the online version of VCG is strategyproof only if the allocation rule is time-monotonic. They also show that a simple fixed price payment can result in a close-to-optimal solution for a non-interruptible version of their problem. 
\citet{lavi-nisan04dynamic-online-auctions} study online supply curves based auction of {\em identical divisible goods} that ensures truthfulness. In this paper, we consider an offline allocation problem but a comparatively more complex one (multiple resources and indivisible tasks of different length). 
% Considering identical goods is similar to a multi-unit auction. However, we examine a setting that requires a solution similar to a multi-unit combinatorial auction. 
% \cite{hajiaghayi2005online} provide an online auction for reusable goods, e.g., assigning time slots of a machine to unit length jobs. They study {\em single-capacity machines} and give a truthful randomized mechanism. 
%They give approximately optimal deterministic and truthful mechanisms for models that allow the unit-length jobs' preemptions. One of the models assumes that the job is not resumable once preempted.
% In Our work is different in the sense that the jobs we consider and the machine capacities can be of varied length.
%
\citet{chen2016efficient} propose a truthful approximate mechanism for online allocation of job to machines where the job can resume or restart once preempted. We provide a comparatively better approximation ratio for efficiency, albeit in an offline setting.
The other related line of work involves designing incentives in queueing problems with specific cost structures that aim to find an efficient allocation truthfully and also ensures budget balance~\cite{mitra2001mechanism,bloch2017second,ghosh2020prior}, while our model can admit costs of any structure.
% However, all of these methods involve monetary transfer as these application domains admit options of transferring money.

In the context of job scheduling without money, \citet{koutsoupias2014scheduling} studies the allocation of independent tasks to machines.
% , where the machines are lazy agents and prefer not to execute any task. 
Every machine reports the time it takes to execute each task and the mechanism provides an approximation to the minimum {\em makespan} in a truthful manner without money for one task---which can be repeated for multiple tasks. 
% However, the question of the approximation ratio in the latter case remains open.
In this paper, we maximize the {\em sum of the visitors' valuations} which are independent of the length of the job and provide an approximate mechanism for multiple tasks maintaining the slot-capacity.
% , and provide an upper bound on the approximation ratio.
% 
\citet{braverman2016optimal} study a similar problem of the allocation of medical treatments at hospitals that have differential costs to patients and the patients value the hospitals differently. 
% A central authority pays all the hospital bills but has a finite budget. 
The waiting time before being admitted to the hospital helps to get a stable matching. However, the value of the agents do not change over slots and hence is different from our setting.
\section{Single-slot Job Allocation Setup}
\label{sec:setup}

Define $N \coloneqq \{1, \ldots, n\}$ to be the set of agents that are trying to access a facility ${\cal F}$. Time is divided into {\em periods}, and each period is further divided into {\em slots}. The set of slots is denoted by $M \coloneqq \{1, \ldots, m\}$. Every slot has a maximum capacity of $k$, which is decided by the region's social distancing norm based on the size of the facility.\footnote{The analysis and results will follow even if the capacity $k_j$ varies with the slots $j \in M$.} A {\em central planner} (e.g., an AI app) allocates these slots to the agents, maintaining the capacity constraint. Every agent $i$ has a {\em cardinal} preference $v_{ij} \in \mathbb{R}_{\geqslant 0}$ (called the agent's {\em valuation}) if her immediate unprocessed job is allocated slot $j$. The valuation implicitly reflects the importance of visiting the facility for that agent. The valuation vector of $i$ is represented by $v_i = (v_{ij}, j \in M) \in \mathbb{R}_{\geqslant 0}^m$.
In this paper, we consider different facilities independently.
The joint facility-slot allocation problem can be modeled as a similar problem with the additional constraint that the same slot cannot be allocated to the same agent at different facilities. We leave the detailed analysis for it as future work.

The planner decides the allocation which can be represented as a matrix $A = [a_{ij}]$, where $a_{ij} = 1$, if agent $i$ is allotted slot $j$, and zero otherwise. We assume that every agent can be assigned at most one slot in a period, and the total number of agents assigned to each slot does not exceed $k$. We denote the slot assigned to $i$ by $a_i^*$. The planner also decides a delay $d = (d_i, i \in N)$, where $d_i$ is the time-delay (in the same unit as the valuation) of agent $i$ before which she cannot make another request to the system. The net payoff of an agent is assumed to follow a standard {\em quasi-linear form}~\cite{SL08}, which implies that every agent wants a more valued slot to be assigned to her and also wants to wait less.
\begin{equation}
    \label{eq:utility}
    u_i((A,d), v_i) = v_i(A) - d_i, \text{ where } v_i(A) = v_{i a_i^*}.
\end{equation}
Denote the set of all allocations by ${\cal A}$. The delays $d_i \in \mathbb{R}_{\geqslant 0}, \forall i \in N$. The planner does not know the valuations of the agents. Therefore he needs the agents to report their valuations to decide the allocation and the delay. This leaves the opportunity for an agent to misrepresent her true valuation. To distinguish, we use $v_{ij}$ for the true valuation and $\hat{v}_{ij}$ for reported valuations. 
In the first part of this paper, we will consider single-slot jobs and use the shorthand $v = (v_i)_{i \in N}$ to denote the true valuation profile represented as an $m \times n$ real matrix, and $\hat{v}$ to denote the reported valuation profile. The notation $v_{-i}$ denotes the valuation profile of the agents except $i$.
The decision problem of the planner is, therefore, formally captured by the following function.
\begin{definition}[Social Scheduling Function (SSF)]
 A {\em social scheduling function (SSF)} is a mapping $f : \mathbb{R}^{m \times n} \to {\cal A} \times \mathbb{R}^n$ that maps the reported valuations to an allocation and delay for every agent. Hence, $f(\hat{v}) = (A(\hat{v}), d(\hat{v}))$, where $A$ is the allocation and $d$ is the delay function.\footnote{We overload the notation $A$ and $d$ to denote both functions and values of those functions, since their use will be clear from the context.}
\end{definition}

\section{Preliminary Definitions}
\label{sec:desiderata}

In this section, we formally define a few desirable properties that a social scheduling function should satisfy. The properties address the issues of prioritization, truthfulness, voluntary participation, and computational complexity.

The first property ensures that the allocation is {\em efficient} in each period, i.e., it maximizes the sum of the valuations of all the agents. 
\begin{definition}[Efficient Per Period (EPP)]
 \label{def:efficiency}
 An SSF $f$ is {\em efficient per period (EPP)} if at every period, it chooses an allocation $A^*$ that maximizes the sum of the valuations of all the agents. Formally, if $f(\cdot) = (A^*(\cdot), d(\cdot))$, then
 \begin{equation}
     \label{eq:efficient}
     A^*(v) \in \argmax_{A \in {\cal A}} \sum_{i \in N} \sum_{j \in M} v_{ij} a_{ij}.
 \end{equation}
\end{definition}

However, since the planner can only access the reported values $\hat{v}_i$'s, which can be different from the true $v_i$'s, it is necessary that the reported values are indeed the true values. The following property ensures that the agents are incentivized to `truthfully' reveal this information {\em irrespective of the reports of the other agents}.
\begin{definition}[Per Period Dominant Strategy Truthful]
 \label{def:non-manipulability}
 An SSF $f(\cdot) = (A(\cdot), d(\cdot))$ is {\em truthful in dominant strategies per period} if for every $v_i, \hat{v}_i, \hat{v}_{-i}$, and $i \in N$
  \begin{align*}
v_i(A(v_i, \hat{v}_{-i})) - d_i(v_i, \hat{v}_{-i}) \geqslant v_i(A(\hat{v}_i, \hat{v}_{-i})) - d_i(\hat{v}_i, \hat{v}_{-i}).
  \end{align*}
\end{definition}

The next property ensures that it is always weakly beneficial for every rational agent to participate in such a mechanism.
\begin{definition}[Individual Rationality]
 \label{def:ir}
 An SSF $f(\cdot) = (A(\cdot), d(\cdot))$ is {\em individually rational} if for every $v$, and $i \in N$
  \begin{align*}
   v_i(A(v)) - d_i(v) &\geqslant 0.
  \end{align*}
\end{definition}

Large facilities that have a large number of high-capacity slots lead to an exponential increase in the size of the set ${\cal A}$. This largeness of ${\cal A}$ makes it challenging to find a solution quickly. 
There are problems for which a quick method of solving them is not known yet. Still, a given solution can be \textit{verified} quickly in time polynomial of the input size, which means that the given solution's validity can be tested quickly. The complexity class of such problems is known as \textit{NP (nondeterministic polynomial)}. We consider the NP-complete class of problems, which is defined as follows.
\begin{definition}[NP-complete]
A decision problem $Q$ is NP-complete if:
\begin{enumerate}
    \item $Q $ is in NP (class of all decision problems verifiable in	polynomial time), and
    \item Every problem in NP is reducible to $Q$ in polynomial time.
\end{enumerate}
\end{definition}
We will show that the slot allocation problem in a certain setting belongs to this class.

In a practical setting, where the allocations and delays need to be decided before every period, it is desirable to have an SSF that is computable in a time polynomial in $n$ and $m$ so that it finishes the computation in a time negligible to the time duration of the period. We consider algorithms that are {\em strongly polynomial} \cite{grotschel1993complexity}. 
% 
% For completeness of the paper, we formally define this class of problems. 
% 
The arithmetic model of computation defines strongly polynomial algorithms. It is assumed that the basic arithmetic operations (addition, subtraction, multiplication, division, and comparison) take a unit time step to perform, regardless of the operands' sizes. 

% [refer to the wikipedia article: https://en.wikipedia.org/wiki/Time_complexity#Strongly_and_weakly_polynomial_time]

\begin{definition}[Strongly Polynomial]
 \label{def:strong-poly}
    An algorithm runs in {\em strongly polynomial time} if \cite{grotschel1993complexity}
    \begin{itemize}
        \item the number of operations in the arithmetic model of computation is bounded by a polynomial in the number of operands in the input instance; and
        \item the space used by the algorithm is bounded by a polynomial in the input size.
    \end{itemize}
\end{definition}
An SSF is strongly polynomial-time computable if there exists an algorithm that computes it in a time strongly polynomial in $n$ and $m$, irrespective of the size of the actual data, such as the value of the $v_i$s or $k$.

\section{Periodic Mechanisms}\label{sec:online}

We consider mechanisms that run at every period of this social scheduling problem. The agents report their valuations at the beginning of every period. The planner decides the schedules and delays.\footnote{For mechanisms that consider the dynamic extension of such allocation problems with finite or infinite horizon~\cite[e.g.]{bergemann-valimaki10dynamic-pivot}, (a)~the designer needs to know the transition probabilities, (b)~equilibrium guarantees are weaker, and (c)~are computationally expensive. These factors made us restrict our attention to periodic mechanisms.}
Since the agents have the opportunity to overstate their importance to get a better slot allotted to them, our approach that uses the ideas of mechanism design \cite{borgers2015introduction} to this social scheduling problem is useful. 
 We use the delay as a surrogate for transferable utility among the agents to satisfy several desirable properties. 
 For the single-slot job setup, our proposed mechanism is as follows.
%
% \begin{algorithm}[h!]
%  \caption{\mech\ in every period}
%  \begin{algorithmic}[1]
%   \STATE {\bf Input:} for every agent $i \in N$, the value $\hat{v}_i$
%       \STATE compute $A^*(\hat{v})$ (\Cref{eqn:LP}) as the allocation \label{alg:allocation-step}
%       \STATE charge a delay $d_i(\hat{v})$ (\Cref{eqn:payment}) to every $i \in N$ for which they cannot access the scheduling mechanism again
%     \STATE {\bf Output:} $A^*(\hat{v})$ and $d(\hat{v})$
%   \end{algorithmic}
%  \label{algo:mechanism}
%  \end{algorithm}

\paragraph{Description of \mech.}
The SSF needs to decide on the allocation $A$ and the delay $d$. 
% The actual allocation that satisfies EPP is given by the following integer program (IP).
% \begin{align}
%  \label{eqn:IP}
%  \begin{split}
%      \argmax_{A} & \sum_{j \in M} \sum_{i \in N} v_{ij} a_{ij} \\ 
%      \text{s.t. } & \sum_{j \in M} a_{ij} \leqslant 1, \ \forall i \in N \\
%      & \sum_{i \in N} a_{ij} \leqslant k, \ \forall j \in M \\
%      & a_{ij} \in \{0, 1\}, \ \forall i \in N, j \in M.
%  \end{split}
% \end{align}
% However, IPs are NP-complete in general. Therefore, 
\mech\ computes the allocation as follows.
\begin{align}
 \label{eqn:LP}
 \begin{split}
     & \argmax_{A} \sum_{j \in M} \sum_{i \in N} v_{ij} a_{ij} \\ 
     & \text{s.t. } \sum_{j \in M} a_{ij} \leqslant 1, \ \forall i \in N; \ \sum_{i \in N} a_{ij} \leqslant k, \ \forall j \in M \\
     & a_{ij} \geqslant 0, \ \forall i \in N, j \in M.
 \end{split}
\end{align}
This is an LP relaxation of the actual allocation problem, which allows $a_{ij}$s to be only in $\{0,1\}$.
We will show that this is without loss of optimality since the solution to LP~(\ref{eqn:LP}) will always be integral and will coincide with the solution of the corresponding IP. 

The delays of agents are computed via the standard VCG payment rule. Denote the optimal allocation of LP~(\ref{eqn:LP}) by $A^*(v)$. Also, denote the allocation given by LP~(\ref{eqn:LP}) when agent $i$ is removed from the system by $A^*_{-i}(v_{-i})$.
For agent $i$, the delay is given by,
\begin{align}
\label{eqn:payment}
\begin{split}
    d_i &:= \sum_{\ell \in N \setminus \{i\}} v_\ell(A^*_{-i}) - \sum_{\ell \in N \setminus \{i\}}  v_\ell(A^*). 
\end{split}
\end{align}
The mechanism in every period is described in \Cref{algo:mechanism}.

\begin{algorithm}[H]
\caption{\mech\ in every period}
 \begin{algorithmic}[1]
  \STATE {\bf Input:} for every agent $i \in N$, the value $\hat{v}_i$
      \STATE compute $A^*(\hat{v})$ (\Cref{eqn:LP}) as the allocation \label{alg:allocation-step}
      \STATE charge a delay $d_i(\hat{v})$ (\Cref{eqn:payment}) to every $i \in N$ for which they cannot access the scheduling mechanism again
    \STATE {\bf Output:} $A^*(\hat{v})$ and $d(\hat{v})$
  \end{algorithmic}
 \label{algo:mechanism}
     \end{algorithm}

% Note that the allocation is different from that of a {\em sorted importance sequential dictator}, i.e., where the agents are sorted based on their importance and asked to pick their favorite slots in that order. For example, consider two agents, A and B, and two slots of capacity one. Valuations are (51, 50) and (50, 0) respectively for A and B respectively for the two slots -- sequential dictator would allocate slot 1 to A and 2 to B, while the socially optimal is to allocate slots 2 and 1 to A and B respectively, which is allocated by \mech.

In the following few sections, we present the theoretical results related to single and multi-slot jobs. %Due to the page limitation, we put the proofs in the appendices available in the supplementary material.

\section{Single-slot Jobs}
\label{sec:theory}
We first show that the allocation given by \mech\ indeed maximizes per-period social welfare.
\begin{theorem}
 \label{thm:relax}
 The allocation of \mech\ given by LP~(\ref{eqn:LP}) always gives integral solutions.
\end{theorem}
\begin{proof}
 Consider the vector $x^\top = (a_{11}, \ldots, a_{1m}, \ldots, a_{n1}, \ldots, a_{nm})$, i.e., the rows of $A$ linearized as a vector. We can write the constraints of LP~(\ref{eqn:LP}) in using a $(n+m) \times nm$ constraint matrix, s.t.,
 \[
 \left ( 
 \begin{array}{ccccccccc}
     1 & \ldots & 1 & 0 & \ldots & 0 & 0 & \ldots & 0 \\
     0 & \ldots & 0 & 1 & \ldots & 1 & 0 & \ldots & 0 \\
     &&&&\ldots&&&& \\
     1 & \ldots & 0 & 1 & \ldots & 0 & 1 & \ldots & 0 \\
     0 & 1 & \ldots & 0 & 1 & \ldots & 0 & 1 & 0 \\
     &&&&\ldots&&&&
 \end{array}
 \right ) x \leqslant \left ( \begin{array}{c}
     1 \\
      \vdots \\
      k \\
      \vdots
 \end{array}\right )\]
 Denote the matrix on the LHS by $C$. The first $n$ and the next $m$ rows correspond to the first and second set of constraints of LP~(\ref{eqn:LP}) respectively. We show that $C$ is totally unimodular (TU), which is sufficient to conclude that LP~(\ref{eqn:LP}) has integral solutions.
 We use the Ghouila-Houri characterization \cite{camion1965characterization} to prove that $C$ is TU. According to this characterization, a $p \times q$ matrix $C$ is TU if and only if each set $R \subseteq \{ 1,2, \cdots, p\}$ can be partitioned into two sets $R_1$ and $R_2$, such that, $\sum\limits_{i \in R_{1}} a_{ij} - \sum\limits_{i \in R_{2}} a_{ij} \in \{ 1,0,-1\}$, for $j= 1,2,\cdots,q$. Note that, in $C$ every column has two 1's, one in the first $n$ rows and one in the next $m$ rows. Pick any subset $R$ of the rows, construct the $R_1$ to be the rows that come from the first $n$ rows, and $R_2$ to be the rows that come from the last $m$ rows (one of these partitions can be empty). Clearly, the difference in each column of the rows $R$ will lie in $\{1,0,-1\}$. Hence proved.
\end{proof}

The result above shows that the optimal solution of LP~(\ref{eqn:LP}) is an optimal solution of the corresponding integer program that maximizes the per-period social welfare. Hence, we conclude the following.
\begin{corollary}
\label{cor:efficient}
 \mech\ is EPP.
\end{corollary}

Even though the LP formulation of \mech\ is without loss of optimality, in general, LPs can be weakly polynomial, i.e., the space used by the algorithm may not be bounded by a polynomial in the size of the input. However, we show that an even stronger result holds for \mech. The forthcoming results show that the allocation and delays of \mech\ are strongly polynomial.
% (\Cref{def:strong-poly}). 
To show this, we will first reduce the allocation problem (LP~(\ref{eqn:LP})) to a minimum weight $b$-matching problem, which is known to be strongly polynomial \cite{schrijver2003combinatorial}.

Consider an edge-weighted bipartite graph $(N, M, E)$, where $N$ and $M$ are the agent set and set of slots respectively. The set $E$ denotes the edges $(i,j)$ with weights $-v_{ij}$. The matching constraints are given by $b_i = 1, \forall i \in N$, and $b_j = k, \forall j \in M$.

\begin{lemma}
\label{lem:equivalence}
 Let $E^* \subseteq E$ be a perfect $b$-matching in $(N, M, E)$ and $A^* = [a_{ij}^*]_{i \in N, j\in M}$ be an allocation where $a_{ij}^* = 1 \Leftrightarrow (i,j) \in E^*$. The matching $E^*$ is a minimum weight perfect $b$-matching iff $A^*$ is an optimal solution to LP~(\ref{eqn:LP}).
\end{lemma}
\begin{proof}
 We prove this via contradiction. Suppose $A^*$ is not an optimal solution to LP~(\ref{eqn:LP}), i.e., there exists $A'$ which satisfies the constraints and yet gives a larger value to the objective function than that of $A$. Hence, $\sum_{j \in M} \sum_{i \in N} v_{ij} a_{ij}' > \sum_{j \in M} \sum_{i \in N} v_{ij} a_{ij}^*$. Consider the edge set $E'$ corresponding to $A'$. Clearly this is a perfect $b$-matching, since $A'$ satisfies the constraints of LP~(\ref{eqn:LP}), and $E'$ gives a lower weight than $E^*$, which proves that $E^*$ is not the minimum weight perfect $b$-matching. The implications can be reversed to obtain the other direction of the proof.
\end{proof}

Note that the delays are calculated by solving an equivalent LP like LP~(\ref{eqn:LP}) with one less agent. Therefore, each of these LPs is strongly polynomial, and the planner needs to solve $n$ of them. The computation of each delay needs the addition of $2(n-1)$ terms and one subtraction. Hence, the number of computations is polynomial in the number of numbers in the input instance, and the space required is polynomial in the input size. Therefore we conclude the following.

\begin{corollary}
\label{cor:delay}
 The computation of the delays in \mech\ is strongly polynomial.
\end{corollary}
Combining \Cref{thm:relax}, \Cref{lem:equivalence}, and \Cref{cor:delay},  we get the following result.

\begin{theorem}
 \label{thm:strong}
  \mech\ provides a combinatorial, strongly polynomial algorithm for computing a social schedule and delays.
\end{theorem}

%%%%%%%%%%%%%%%%%%%%%%%
%%%% VCG standard part
%%%%%%%%%%%%%%%%%%%%%%%

Since \mech\ uses the VCG payment expression to compute the time delay and because the allocated slots are {\em goods} to the agents, the following two facts follow from the known properties of the VCG mechanism.
\begin{fact}
\label{thm:dsic}
 \mech\ is per period dominant strategy truthful.
\end{fact}
\begin{proof}
 This proof is a standard exercise in the line of the proof for Vickery-Clarke-Groves (VCG) mechanism \cite{Vick61,Clar71,Grov73}. 
 
 Let us assume for the contradiction that, there exist an agent $i$ for having true valuations for the slots as, $v_{i}$, but misreports it as $v{'}_{i}$(the corresponding value function is $v{'}_{i}$), and gets better utility.
 Suppose ${A}(v{'}_{i},v_{-i}) = A'$ and ${A}({v}_{i},v_{-i}) = A^{*}$. The utility of $i$ for $A'$ is:
 \begin{align*}
     \lefteqn{v_i(A')- d_i(v'_{i},v_{-i})} \\
     & = v_i(A') - \sum_{\ell \in N \setminus \{i\}} v_\ell(A(v_{-i})) + \sum_{\ell \in N \setminus \{i\}}  v_\ell(A^{'})\\
     & = \sum_{\ell \in N }  v_\ell(A^{'}) - \sum_{\ell \in N \setminus \{i\}} v_\ell(A(v_{-i}))  
 \end{align*}
 Similarly, the utility of $i$ for $A^{*}$ is:
 \begin{align*}
      = \sum_{\ell \in N }  v_\ell(A^{*}) - \sum_{\ell \in N \setminus \{i\}} v_\ell(A(v_{-i}))  
 \end{align*}
 If $i$ gets better utility by misreporting her valuation as $v'(.)$, then
 \begin{align*}
     \sum_{\ell \in N }  v_\ell(A^{'})  > \sum_{\ell \in N }  v_\ell(A^{*}) 
 \end{align*}
 The above inequality leads to the contradiction that $A^{*}$ is optimal for the reported valuation $({v}_{i},v_{-i})$. Therefore, \mech\ is dominant strategy truthful in every period.
\end{proof}

\begin{fact}
 \label{thm:ir}
 \mech\ is individually rational for every agent.
\end{fact}
\begin{proof}
 Consider agent $i$. The utility of $i$ under \mech\ is $v_i(A^*(v)) - d_i(v)$
 \begin{align*}
     &= v_i(A^*) - \Big( \sum_{\ell \in N \setminus \{i\}} v_\ell(A^*_{-i}) - \sum_{\ell \in N \setminus \{i\}}  v_\ell(A^*) \Big) \\
     &= \Big( \sum_{\ell \in N}  v_\ell(A^*) - \sum_{\ell \in N} v_\ell(A^*_{-i}) \Big) + v_i(A^*_{-i}) \geqslant 0
 \end{align*}
 The second equality holds by reorganizing the terms and adding and subtracting $v_i(A^*_{-i})$. Note that the difference term in the parentheses in the last expression is always non-negative since $A^*$ is the optimal allocation for all allocations. In particular, $A^*_{-i}$ is also a feasible allocation when agent $i$ is present. The term $v_i(A^*_{-i})$ is zero. Hence the inequality follows.
\end{proof}
% The proof sketches of the above facts are available in \Cref{APP_DSIC,APP_IR}.

The multi-slot jobs, unlike the single-slot jobs, are relatively difficult to schedule, as we discuss in the following section.

\section{Multi-slot Jobs}
\label{sec:multi}

In this section, we consider jobs with different lengths, i.e., for agent $i$, the job may be of length $l_i \geqslant 1$. Since the job is {\em indivisible}, the entire length $l_i$ of the job requires contiguous slots for execution within the period. 
% Examples of such jobs are computer programs that are not parallelizable and need continuous execution. 
For example, an individual may visit a facility (e.g., a shopping mall) for a quick shopping, which may take a shorter duration, or for dining, which may take longer. However, all these jobs are indivisible, and the allocation needs to provide contiguous time-slots to that agent. The agents report the valuations and lengths of their jobs. We show that the optimal allocation problem in such a setting can be computationally intractable. The notation is mildly updated as follows to accommodate the multi-slot jobs.

Each agent $i$ gets a valuation $v_{ij}$ for her last unprocessed job if her job begins at slot $j$, and has a length $l_i$. The value of the job is zero if (a)~it starts at any of the last $(l_i - 1)$ time-slots of the period (since it cannot finish within the period), and (b)~if the job is unallocated.

A matrix $V$ consists of the agents' reported valuations, and $L$ consists of the lengths of agents' jobs. Allocation is given by the matrix $\AA=[\aaa_{ij}]$, where $\aaa_{ij}=1$ if agent $i$'s job starts at slot $j$, else $\aaa_{ij}=0$, and $\aaa_{i}$ represents the slot allocation vector for agent $i$.  
% Let $q_j$ represents the number of job served at slot $j$. For allocation $\AA$, $q_j=\sum_{i \in N} \sum_{p \in M,\  j \leqslant p+l_i-1} \aaa_{ip}$. 
Keeping all other notations as before, we define the \multi\ problem as follows.
\begin{definition}[Multi-slot Indivisible jobs Allocation problem (\multi)]: Given $(N, M, V, L,k)$, find an allocation $\AA$, such that $\sum_{i \in N} \sum_{j \in M} v_{ij}(\aaa_{ij})$ is maximum, subject to the constraints that the total number of jobs allocated in a slot does not exceed the capacity of the slot, and each job $i$ is assigned to at most $l_i$ contiguous slots. Mathematically, \multi\ is given by the following integer program (IP).
\begin{align}
 \label{eqn:MULTIALLOC_IP}
 \begin{split}
      & \argmax_{A} \sum_{i \in N} \sum_{j \in M}  v_{ij}\ \aaa_{ij} \\
      & \text{s.t. } \sum\limits_{i \in N} \sum\limits_{\substack{p \in M \\ j \in [p, p+l_i-1]}} \aaa_{ip} \leqslant k, \forall j \in M,\\  &  \sum_{j \in M} \aaa_{ij} \leqslant 1, \forall i \in N; \ \aaa_{ij}\in \{ 0,1\}, \forall i \in N, \forall j \in M
 \end{split}
\end{align}
\end{definition}
% The LP relaxation of the IP~\ref{eqn:MULTIALLOC_IP} is identical to LP~(\ref{eqn:LP}) without the first set of constraints and hence, the constraint matrix of new LP is also TU. Therefore, the new LP also gives integral solutions which is an optimal solution to the IP~\ref{eqn:MULTIALLOC_IP}.
% We denote the decision version of \multi\ with the same name, which is defined as follows:
% \begin{quote} % looks better
% Given $(N,M,V,L,k,\beta)$, is there an allocation $\AA$ such that $\sum_{i \in N} \sum_{p \in M,\  j \leqslant p+l_i-1} \aaa_{ip} \leqslant k \ \forall j \in M, \ \sum_{j \in M} \aaa_{ij} \leqslant 1\  \forall i \in N$, and $\sum_{i \in N} \sum_{j \in M}  v_{ij}\ \aaa_{ij}\geqslant \beta$ where $\beta\in\mathbb{R}$.
% \end{quote}
In the first set of inequalities, we sum over every job $i \in N$ and check if it is under execution at $j$, for every $j \in M$. A job $i$ is under execution at slot $j$, if it is allocated at a slot $p$ s.t. $j \leqslant p+l_i-1$. The second set of inequalities ensure that no job is allocated more than once. 
% By the definition of valuation, which is $0$ if the allocation is within the last $l_i-1$ slots, the job is the same as unallocated.

% Depending on the opening hours of the facility and the duration of each slot, $|M|$ can be a small integer. When the input variables such as the number of slots and number of customers are small, \multi\ can be solved in polynomial time by simply checking for every possible solution and choosing the optimal among them. However, for completeness and to avoid situations where our results can not be applied, we consider the computationally expensive conditions for analysis. For example, the size of the period can be significantly longer; hence, $|M|$ can be a large number.

We show that \multi\ is computationally intractable by performing a polynomial reduction from the Multi-Unit Combinatorial Auction (\MUCA), which is NP-complete.
% ~\cite{cramton2004combinatorial,rothkopf1998computationally}.

\smallskip \noindent
{\em Description of \MUCA}:
Consider a multiset $\MM=(\GG, y)$, where $\GG =\{1,2,3, \ldots, g\}$ is a set of goods and $y$ is a function, $y: \GG \to \mathbb{Z}_{\geqslant 0}$ representing the multiplicity or the number of available units of the elements of $\GG$ in $\MM$. Each agent $i \in \NN=\{1,2, \ldots ,n \}$ is a multi-minded bidder, which means $i$ has a positive valuation $w_i(\cdot)$ for multiple bundles of available goods. We call the set of bundles for which agent $i$ has a positive valuation to be the {\em demand set} of $i$, represented by $\DD_i$. The valuation function is such that, $w_i({q})\in \mathbb{R}_{\geqslant 0}, \forall {q} \in \DD_{i}$. We use the following notation $\DD=[\DD_i]_{i \in \NN}$ and, $W_i=(w_i({q}))_{{q} \in \DD_i}$, $W=[W_i]_{i \in \NN}$. In this paper we assume that, every agent demands at most one unit of every good. With this assumption, an allocation of a bundle of goods to the agents is represented as a matrix $\SS=[\sss_{iq}]$, where $\sss_{iq}=1$, if the bundle $q\in \DD_i$ is allocated to $i$, else $\sss_{iq}=0$. 
% The vector $\sss_{i}=[\sss_{ij}]_{j\in\GG}$ and $\SS_{i}$ denotes the set of goods allocated to $i$.
For an allocation $\SS$, every agent $i$ gets a valuation, $w_{i}(\SS)$=$\sum_{q \in \DD_i}w_{i}({q})\ \sss_{iq}$, otherwise $w_{i}(\SS)=0$. The formal definition is as follows.

\begin{definition}[Multi-Unit Combinatorial Auction (\MUCA)]
Given $(\NN,\MM,W,D)$, find an allocation $\SS$ of goods to the agents such that $\sum_{i \in \NN} \sum_{q \in \DD_i}w_{i}({q})\ \sss_{iq}$ is maximum, and the total units of good $j\in \GG$  allocated to the agents does not exceed $j$'s availability $y(j)$, and every agent $i$ is assigned at most one of the demanded bundle from $\DD_{i}$. Mathematically, \MUCA\ is given by the following integer program (IP):
\begin{align}
 \label{eqn:MUCA}
 \begin{split}
    & \argmax_{\SS} \sum\limits_{i \in \NN} \sum\limits_{q \in \DD_i}w_{i}({q})\ \sss_{iq} \\
     & \text{s.t. } \sum_{i \in \NN} \sum\limits_{\substack{q \in \DD_i \\ j \in q}} \sss_{iq} \leqslant y(j), \forall j \in \GG\\
     & \sum_{q \in \DD_i} s_{iq} \leqslant 1, \forall i \in \NN; \ \sss_{iq}\in \{ 0,1\}, \forall i \in \NN, \forall q \in \DD_{i}
 \end{split}
\end{align}
\end{definition}
The reduction of \multi\ to \MUCA\ proceeds as follows.
For a given instance $(N, M, V, L, k)$ of \multi, construct an instance of \MUCA$(\NN, \MM, W, \DD)$ problem such that, the set of agents $\NN$ is $N$, the set of goods $\GG$ is the set of the slots $M$ within the period, where $y(j)=k$, $\forall j \in M$. For every $i\in \NN$, the demand set $\DD_i$ consists of $(m-l_i+1)$ distinct bundles. Each of the bundles in $\DD_i$ is of size $l_i$ and consists of $l_i$ contiguous slots. We denote a bundle as $q_j$ if it contains $l_i$ contiguous slots starting from slot $j$, and $w_{i}(q_j)$ is equal to $v_{ij}$ (the valuation $i \in N$ gets if her job starts at slot $j\in M$). The above construction is done in polynomial steps of the input size. 
We construct a solution of \multi\ from a solution of \MUCA\ in the following way: for every $q_{j}\in \DD_{i}$ and $i\in \NN$, if $\sss_{iq_{j}}=1$ then, $\aaa_{ij}=1$ for every $ i \in N$ and $j\in M$. Similarly, we construct a solution of \MUCA\ from a solution \multi\ in the following way: if $\aaa_{ij}=1$ for $ i \in N$ and $j\in M$ then, $\sss_{iq_{j}}=1$ for every $q_{j}\in \DD_{i}$ and $i\in \NN$. The following lemma shows that an optimal solution of \multi\ is an optimal solution of \MUCA\ and vice-versa.

% \begin{lemma}\label{lem:MUCA_reduction}
% Let $\SS^*$ is an optimal solution for \MUCA\ for a multiset of goods $\MM$, and $\AA^*$ is such that, $\aaa_{ij}^{*}=1$ for $i \in N$ and $j \in M$, if and only if $\sss_{iq_{j}}^{*}=1$ in $\SS^*$ for $i \in \NN$ and $q_{j} \in \DD_{i}$, then $A^*$ is an optimal solution for \multi.
% \end{lemma}
\begin{lemma}\label{lem:MUCA_reduction}
Let $\SS^*$ is a solution for \MUCA\ for a multiset of goods $\MM$, and $\AA^*$ is such that, $\aaa_{ij}^{*}=1$ for $i \in N$ and $j \in M$, if and only if $\sss_{iq_{j}}^{*}=1$ in $\SS^*$ for $i \in \NN$ and $q_{j} \in \DD_{i}$, then $A^*$ is an optimal solution for \multi\ if and only if $\SS^*$ is an optimal solution for \MUCA.
\end{lemma}
\begin{proof}
Suppose the above statement is not true and hence $\AA^{'}$ but not $\AA^{*}$ is an optimal solution for \multi.
\begin{equation*}
     \sum_{i \in N}  \sum_{j \in M} v_{ij}\ \aaa'_{ij} \geqslant \sum_{i \in N} \sum_{j \in M}   v_{ij}\ \aaa^{*}_{ij}
\end{equation*}
As $w_{i}(q_j)=v_{ij}$, and $\aaa^{*}_{ij}=1$ only if $\sss^{*}_{iq_{j}}=1$, with the constructed $\SS^{'}$ corresponding to $\AA^{'}$ the following inequality holds,
\begin{equation*}
    \sum_{i \in \NN}  \sum_{q_{j} \in \DD_{i}}  w_{i}(q_j)\ \sss^{'}_{iq_{j}} \geqslant  \sum_{i \in \NN} \sum_{q_{j} \in \DD_{i}}  w_{i}(q_j)\ \sss^{*}_{iq_{j}}
\end{equation*}
The above equation results in a contradiction that $\SS^*$ is an optimal solution for \MUCA.

Since each step of the above proof has implications in both directions, the other direction of the proof is implied.
% To prove the other direction, we construct $\AA^{*}$ for \multi\ from the optimal solution $\SS^{*}$ for \MUCA\ as explained before the \Cref{lem:MUCA_reduction}. The proof follows by similar argument in the reverse order.
\end{proof}

Since \MUCA\ is NP-complete \cite{cramton2004combinatorial,rothkopf1998computationally}, using \Cref{lem:MUCA_reduction}, we get the following theorem.
\begin{theorem}\label{thm:MULTI_NPC}
\label{thm:NPc}
 \multi\ is NP-complete.
\end{theorem}

% The above theorem is proved using the reduction from the \textsc{knapsack problem} and the proof is available in \Cref{NPC_proof}.

% Though finding the efficient allocation for \multi\ is intractable, 
However, it is possible to find an approximately efficient allocation in polynomial time that is truthful and individually rational.
% In the following, we present a truthful polynomial time algorithm to achieve $O(k\ m^{\frac{1}{k-2}})$ approximation to the optimal solution for \multi\ problem, by reducing it to the Multi-Unit Combinatorial Auction (\MUCA) problem. 
% There exists a polynomial time truthful algorithm to achieve $O \big(k  g^{\frac{1}{k-2}} \big) $ approximation to optimal solution for the \MUCA\ \cite[Theorem 3]{bartal2003incentive}. 
To find that, we leverage the approximation algorithm of \MUCA\ due to \cite[Theorem 3]{bartal2003incentive}.
Using \Cref{lem:MUCA_reduction} and the next few results, we prove that there exists a polynomial time truthful mechanism (\maafull\ or \MAA) to achieve $O \big(k m^{\frac{1}{k-2}} \big)$ approximation to the optimal solution of \multi.

% \begin{wrapfigure}{l}{0.52\textwidth}
% %
% \begin{minipage}{0.52\textwidth}
%
\begin{algorithm}[t]
\caption{\MAA\ in every period}\label{alg:alg2}
\begin{multicols}{2}
\begin{algorithmic}[1]
\STATE \textbf{Procedure} {\MAA}{($N, M, V, L, k$)}
\STATE $b \gets \underset{i}{\arg} \max\limits_{i \in N, j\in M}v_{ij}$; $\ v_{\max} \gets \max\limits_{i \in N, j\in M}v_{ij}$; \\ $\ r \gets (6m(k-1))^{\frac{1}{k-2}}$
\STATE $\QQ^{0} \gets [0, 0, \ldots, m \ \text{times}]$
% \FOR{$j = \{1, 2, \ldots, m\}$}
% \STATE $P^{0}_{j} = \frac{v_{\max}}{6m(k-1)} =: \pi_0$
% \ENDFOR
\STATE $ \aaa_{bs'}^{*} =1,\ \text{where} \ s'\gets \underset{j \in M}{\argmax} (v_{bj})$; \\ $ \aaa_{bj}^{*} =0$,\ $\forall j\in M\setminus\{ s'\}$; $\ \PP_{b}=v_{\max}^{-b}$
\FOR{$i = \{1, 2, \ldots, n\}$ and $i \neq b $}
% \STATE $v_{\max}^{-i} \gets \underset{t \in N \setminus \{i\}, j \in M}{\max} v_{tj}$
% \IF {$i \neq b $} 
\FOR{$j = \{1, 2, \ldots, m\}$} 
\STATE $P^{i}_{j} \gets \pi_0 \cdot r^{Q^{i-1}_{j}}$ % P^{0}_{j} is \pi_0 for all j
\ENDFOR 
\STATE $ \aaa_{is}^{*} =1,\ \text{where} \ s\gets \text{\textbf{\texttt{max}}}(P^{i}, v_{i})$\  (\Cref{eq:max})\\ % s is well defined in this iteration, hence one can directly refer to it.
\STATE $ \aaa_{ij}^{*} =0$,\ $\forall j\in M\setminus\{ s\}$
\STATE $\mathcal{P}_{i} \gets \sum_{j = s}^{s+l_i-1} P^{i}_{j}$
\FOR{$j = \{1, 2, \ldots, m\}$}
\IF {$j \in [s, s+l_i-1]$}
\STATE $\QQ^{i}_{j} \gets \QQ^{i-1}_{j} + 1 $ 
\ELSE
     \STATE $\QQ^{i}_{j} \gets \QQ^{i-1}_{j}$
\ENDIF
\ENDFOR
% \ELSE 
% \STATE $ \aaa_{bs'}^{*} =1,\ \text{where} \ s'\gets \underset{j \in M}{\argmax} (v_{bj})$; \\ $ \aaa_{bj}^{*} =0$,\ $\forall j\in M\setminus\{ s'\}$; $\ \PP_{b}=v_{\max}^{-b}$
% \ENDIF
\ENDFOR
\RETURN $\aaa^{*}, \mathcal{P}$
\end{algorithmic}
\end{multicols}
\end{algorithm}
    % \end{minipage}
%   \end{wrapfigure}
% 

The operational principle of \MAA\ is a sequential dictatorship, where the sequence is an arbitrary order (WLOG $1,2,\ldots, n$) of the agents and is independent of the information submitted by them. The mechanism comes with a price\footnote{The terms {\em price} and {\em delay} are equivalent in the rest of the paper.} vector which is updated while iterating over the agents in the sequence. We use a superscript $i$ to denote the the price faced by the agent $i$ for slot $j$, $P_{j}^{i}$, when $i$'s turn comes. Hence, $P^{i}=[P_{j}^{i}]_{j\in M}$ denotes the price vector seen by $i$. The mechanism also uses a function \textbf{\texttt{max}} that returns the slot $s$ that maximizes agent $i$'s utility given her valuation vector $v_i$ and the price vector $P^{i}$ when her job of length $l_i$ starts from slot $s$. Mathematically, \textbf{\texttt{max}} is defined as follows:
\begin{equation}
  \textbf{\texttt{max}}(v_{i}, P^{i}) =  \argmax_{s\in M}\  (v_{is}-\sum\limits_{j\in[s,s+l_i-1]} P_{j}^{i})
  \label{eq:max}
\end{equation}
\MAA\ maintains a vector  $\QQ^{i}=[\QQ_{j}^{i}]_{j \in M}$, where $\QQ_{j}^{i}$ denotes the current allocated population of the slot $j$ after allocating slots to $i$.
First, \MAA\ picks the agent $b$ that has the maximum valuation $v_{\max}$ for any slot, and initializes $\QQ^{0}_j=0, \forall j\in M$. The initial price of every slot is set to $\pi_0 := \frac{v_{\max}}{6m(k-1)}$ and a constant factor $r=(6m(k-1))^{\frac{1}{k-2}}$ is defined. Consider an arbitrary order (WLOG $(1,2,\ldots, n)$) of the agents. For each agent $i \neq b$ in sequence, the price $P^{i}_{j}$ is computed using $\QQ^{i-1}_j$, $r$, and $\pi_0$ such that the prices of the slots increase by a multiplicative factor of a suitable exponent of $r$ such that the prices for more congested slots are higher. Then the highest utility-deriving slot $s$ to start $i$'s job is found using \textbf{\texttt{max}}, and the corresponding allocation vector for $i$ is represented as $\aaa_{i}^{*}$, where $\aaa_{is}^{*} = 1, \aaa_{ij}^{*} = 0, \forall j \neq s$. The total price (or delay) charged to $i$ is denoted by $\PP_{i}$ and is equal to $\sum_{j\in [s,s+l_i-1]} P^{i}_{j}$. The vector $\QQ^i$ is updated after the allocation of slots to agent $i$. The agent $b$ gets her most valued starting slot and pays the maximum valuation among all other agents and slots, which is represented by $v_{\max}^{-b}$. 

An important feature of \Cref{alg:alg2} is that it {\em does not} explicitly check the capacity constraint. However, we show that the choices of $\pi_0$ and $r$ implicitly maintains that in the following result. The units of slot $j$ after \Cref{alg:alg2} executes that are occupied by agents except $b$ are $\QQ_{j}^{*}$.
% and is similar to that given by \citet[Lemmata 2 and 7]{bartal2003incentive}, where $P^{0}$ and $r$ are as defined in \Cref{alg:alg2}.
\begin{lemma}\label{lem:alg2_validity}
Let $\pi_0, r, \delta > 0$ be such that $\pi_0 r^{\delta} \geqslant v_{\max}$, then $\QQ_{j}^{*} \leqslant \delta+1$. This implies that \MAA\ maintains the capacity constraints for each slot for $\delta=k-2$.
\end{lemma}
\begin{proof} Assume for contradiction that $\QQ_{j}^{*} > \delta + 1$ and let $i$ be the first customer due to which this contradiction takes place for some slot $j$, i.e., $\QQ_{j}^{i} > \delta + 1$. Since each customer does not get more than one unit of any slot, then it must be that $\QQ_{j}^{i-1} > \delta$. Hence, for slot $j$, the following holds: $P_{j}^{i} > \pi_0r^{\delta} \geqslant v_{\max} \geqslant \max_{j \in M}\  v_{ij} $. This makes $i$'s total price for $l_{i}$ contiguous slots including slot $j$ to be more than her corresponding total valuation for those slots. This contradicts the definition of \textbf{\texttt{max}} since the utility becomes negative for agent $i$.
\end{proof}
The allocation to $b$ is at most {\em one} unit from each slot. With carefully choosing $\pi_0$ and $r$, we bound the units of any slot allocated to all the other agents, $\QQ_{j}^{*}$ to $(k-1)$, maintaining the possibility of the maximum use of every slot.

Since \textbf{\texttt{max}} allocates a slot only if that allocation increases the agent's utility, the following result holds.
\begin{theorem}\label{lem:alg2IR}
\MAA\ is individually rational.
\end{theorem}
Next, we prove that misreporting the private information ($v_{i},l_{i}$) is never beneficial for any agent $i$. 
\begin{theorem}\label{lem:alg2IC}
In \MAA, reporting $v_i$ and $l_i$ truthfully in every period is a dominant strategy for all $i \in N$. 
\end{theorem}
\begin{proof}
For the agent $b$, we see that the utility is that of a second price auction and is independent of its length report. Since for second price auction revealing valuation truthfully is a dominant strategy therefore truthfully revealing valuation and length is a dominant strategy for $b$.

For the other agents, note that \MAA\ considers the agents sequentially and allocates the utility-maximizing available slots in their turn. The order of the agents in \MAA\ is independent of the valuations and lengths of the jobs. Consider agent $i$. When her turn comes, the mechanism picks the slots that give the maximum difference between the valuation of $i$ for those slots and the current prices of those slots. Note that, the prices of those allocated slots are not dependent on the valuation or length reported by agent $i$ (rather it is dependent on the reports of the previous agents in the sequence and $b$), and the mechanism allocates her the optimal set of slots. Hence, by misreporting the valuation $v_i$, agent $i$ can either continue to get the same slots or get a worse set of slots w.r.t.\ her true valuation. Hence, there is no incentive for $i$ to misreport her valuation.

{\em Misreporting length}: If $\hat{l}_{i}<l_i$, \MAA\ allocates only $\hat{l}_{i}$ number of contiguous slots to $i$ (which can have zero value as $\hat{l}_{i}$ is not sufficient for completion of her job) and $i$ can get a negative payoff as she has to pay $\PP_i$ (which is non-negative). If $\hat{l}_{i}>l_i$, then \MAA\ allocates more slots than $i$ actually needs. This allocation does not increase agent $i$'s valuation, but increases the price since now she will be charged for $\hat{l}_{i}$ slots which is larger than the true length.

Combining the above two arguments that hold for all $i \in N$ irrespective of the reports of the other agents, we get the claim.
\end{proof}

To find the best slot for an agent, \textbf{\texttt{max}} checks the feasibility constraints and computes the allocation considering the valuation and the current price of the slots. As there are $(m-l_i+1)$ possible allocations, \textbf{\texttt{max}} requires at most $O(m)$ time for every agent $i$. Therefore, the following result on the complexity of \MAA\ holds.
\begin{theorem}\label{lem:alg2complexity}
\MAA\ has time complexity $O(mn)$.
\end{theorem}

To show the approximation factor of \MAA, we need a few more results. %Due to paucity of space, we present those results (similar to the ones by \citet{bartal2003incentive}, modified according to \multi) in the appendix. These results help us prove the main theorem of this section.
We restate a few results from \cite{bartal2003incentive}, which help us prove the main theorem of \Cref{sec:multi}. The lemma and section numbers of these results in the original paper are mentioned within parentheses in the restated lemmata.

\begin{lemma}[{\cite[Section 4.2, Lemma 4]{bartal2003incentive}}]
\label{lem:bartal}
For every agent $i$, $v_{i}(\aaa^{*}_{i}) \geqslant v_{i}(\aaa^{'}_{i}) - \sum_{ j\in [s, s+l_i-1]\ s.t.\
\aaa'_{is}=1} P^{*}_{j}$ for every allocation $\aaa^{'}_{i}$, where, ${P}^{*}$ is the vector of prices of slots at the end of \Cref{alg:alg2}.\footnote{With a slight abuse of notation, we denote $[a, b]$ to be the integers between $a$ and $b$, where $a<b$.}
\end{lemma}
% \gs{the text in the next sentence is same as the above lemma. Isn't it?}
% From the above result, we get that $\forall i \in N$ and any corresponding customer's allocation $\aaa'_{i}$, $v_{i}(\aaa^{*}_{i}) \geqslant v_{i}(\aaa'_{i}) - \sum_{ j\in [s, s+l_i-1]\ s.t.\
% \aaa'_{is}=1} P^{*}_{j}$, where ${P}^{*}$ is the vector of prices of slots at the end of \Cref{alg:alg2}. 
Let $V(ALG)$ and $V(OPT)$ denotes the sum of valuations of customers for the allocation $\AA^{*}$ given by \MAA, and that for the optimal allocation (say $\hat{\AA}$) respectively. Similarly, $V({ALG}^{-b})$ and $V({OPT}^{-b})$ represents the sum of valuations of every agent except $b$ according to $\AA^{*}$ and $\hat{\AA}$ respectively. Summing it for all the agent $i \in N$, we get the following corollary.
\begin{corollary}\label{cor:alg2_P*}
$V({ALG}^{-b}) \geqslant V(OPT^{-b}) - (k-1) \sum\limits_{j\in M} P^{*}_{j}$
\end{corollary}
The following result provides a lower bound on $V({ALG}^{-b})$. 
\begin{lemma}[{\cite[Section 4.2, Lemma 5]{bartal2003incentive}}] \label{lem:lem5ofbartal}
$V({ALG}^{-b}) \geqslant \frac{\sum_{j \in M} P^{*}_j-m\pi_{0}}{r-1}$
\end{lemma}
Combining \Cref{lem:lem5ofbartal} and \Cref{cor:alg2_P*}, we state the following result.
\begin{lemma}\label{lem:alg2_finalEq}
If $m(k-1)\pi_{0} \geqslant \frac{V(OPT^{-b})}{2}$, then $2((k-1)(r-1)+1) \geqslant V(OPT^{-b})/V({ALG}^{-b})$.
\end{lemma}
Following the conditions in \Cref{lem:alg2_validity} and \Cref{lem:alg2_finalEq}, we fix $\pi_{0} =  \frac{v_{\max}}{6m(k-1)}$, $r= (6m(k-1))^{\frac{1}{k-2}}$. We restate the following result about the approximation ratio from \citet{bartal2003incentive}.
\begin{lemma}[{\cite[Section 5, Lemma 8]{bartal2003incentive}}]\label{lem:lem8ofbartal}
The approximation ratio of \Cref{alg:alg2} is $3((k-1)(r-1)+1)$.
\end{lemma}
Finally, combining \Cref{lem:alg2_finalEq,lem:lem8ofbartal}, we get \Cref{thm:apprx_multi}.

\begin{theorem}
\label{thm:apprx_multi}
There exists a polynomial time $(O(mn))$, incentive compatible, and individually rational mechanism to achieve $O(k m^{\frac{1}{k-2}})$ approximation to optimal solution for \multi.
\end{theorem}

The result above shows the existence of an approximately efficient mechanism that satisfies the other three desirable properties. The question of finding a lower bound on the approximation ratio remains open.

\section{Multi-slot Divisible jobs}
\label{sec:div-jobs}
A {\em divisible} job implies that the job can be broken up into multiple pieces of unit slot-size and can be executed independently within a period. Partial execution of the jobs is also admissible. Examples of such jobs are parallel computer programs, with each piece being an independent thread of the program. In the context of social scheduling, divisible jobs can be interpreted as independent needs of a customer that may be completed via multiple separate visits to the store without any extra cost. The valuation of the job is the sum of the valuations of the allocated slots. We assume that the length $l_i$ of the job is verifiable and is common knowledge. Hence, the valuations of each of the slots, i.e., $v_{ij}$'s, are the agents' only private information. In this setting, the valuation of agent $i$ from the allocation $A$ can be written as $v_i(A) = \sum_{j : a_{ij} = 1} v_{ij} a_{ij}$.
Note that the EPP condition under this setting is identical to the LP~(\ref{eqn:LP}) with the first set of constraints replaced by $\sum_{j \in M} a_{ij} \leqslant l_i, \ \forall i \in N$. However, this does not alter the constraint matrix $C$, which is TU (shown in the proof of \Cref{thm:relax}). Therefore the new LP also has integral optimal solution. 
The {\em modified} \mech\ is identical to \Cref{algo:mechanism} with Step~\ref{alg:allocation-step} replaced with the solution to the modified LP as explained above.
Hence, the following results hold similarly.

\begin{proposition}
 \label{thm:multi-divisible-dsic}
 Modified \mech\ is per period dominant strategy truthful for multi-slot divisible jobs.
\end{proposition}

\begin{proposition}
 \label{thm:multi-divisible-ir}
 Modified \mech\ is individually rational for every agent.
\end{proposition}

\begin{theorem}
 \label{thm:multi-divisible-relax}
 The allocation of the modified \mech\ always gives integral solutions.
\end{theorem}
Hence, the following proposition also holds.
\begin{proposition}
\label{cor:multi-divisible-efficient}
 Modified \mech\ is EPP.
\end{proposition}
To reduce the allocation problem to a perfect $b$-matching problem, we need to alter $b_i = l_i$, and \Cref{lem:equivalence} holds. Therefore, we conclude the following result.
\begin{theorem}
 \label{thm:multi-divisible-strong}
  Modified \mech\ provides a combinatorial, strongly polynomial algorithm for computing a social schedule and delays.
\end{theorem}

\section{Experiments}
\label{sec:experiment}

While the mechanisms presented in the previous sections satisfy several desirable properties of a social scheduling mechanism for indivisible single and multiple slot jobs, its prioritizing profile for different classes of importance, costs of prioritization, and reduction in social congestion are not theoretically captured. In this section, we investigate these properties using real and synthetic datasets. The real dataset we collected from a store gave us only the checkout times. In absence of the length information of the visits, we resorted to the single-slot job model (with hourly slots) and tested the performance of \mech\ on this data (\S\ref{sec:reduction}). For multi-slot jobs, we simulated \MAA\ to find the suboptimality and the reduction in the running time (\S\ref{sec:approx-plot}).
For these experiments, we consider three discrete levels of valuations of the agents denoted by 3, 2, and 1, which can be interpreted as \hi, \med, and \lo\ respectively. The numbers represent the agent's valuation if they are allocated their most preferred slot. The valuations for the other slots are assumed to decrease with a multiplicative factor $\delta \in (0, 1)$ in the order of their slot preference, i.e., the valuation for the $t$-th most preferred slot of a \med\ agent will be $2 \delta^{t-1}$. We used Gurobi~\cite{gurobi} under an academic license for all experiments.
% The agents in this setting are asked to report their slot preference order and what importance their work is. 
% \sn{a little detailed overview of the section when it is complete}

\subsection{Reduction in the social congestion}
\label{sec:reduction}
This section considers a real data of customer footfall in a general store that we have collected from the store. The dataset contains the customers' hourly checkout (billing) time from 7 AM to 9 PM (opening hours of the store) for the whole month of July 2020. Since the dataset was anonymized for customer identification, we have assumed that the billing timestamps are unique users for a day. 
% The approximate dimension of the store is $30 \times 30$ square feet. With 6 feet being the minimum distance for safe social distancing, the recommended number of people in the store at any time should roughly be 16 at any given time. If every shopper takes roughly 30 minutes to shop (including billing and checkout), then in an hour, 
Given the size of the store, around 32 people an hour should be a fair capacity to maintain social distance. However, the data shows that
% the average hourly population of the store was much above this, particularly during the rush hours (5-8 PM) and on weekends -- 
the monthly average during the periods 5-6 PM, 6-7 PM, and 7-8 PM were  38.00, 48.63, and 52.83 respectively. Interestingly, the monthly average of the population in an hour is 26.5, which is well within the safety limits. Therefore, this dataset works as a perfect example where users can benefit significantly from social scheduling.

\begin{figure}[h]
\begin{multicols}{2}
\centering
    \includegraphics[width=\linewidth]{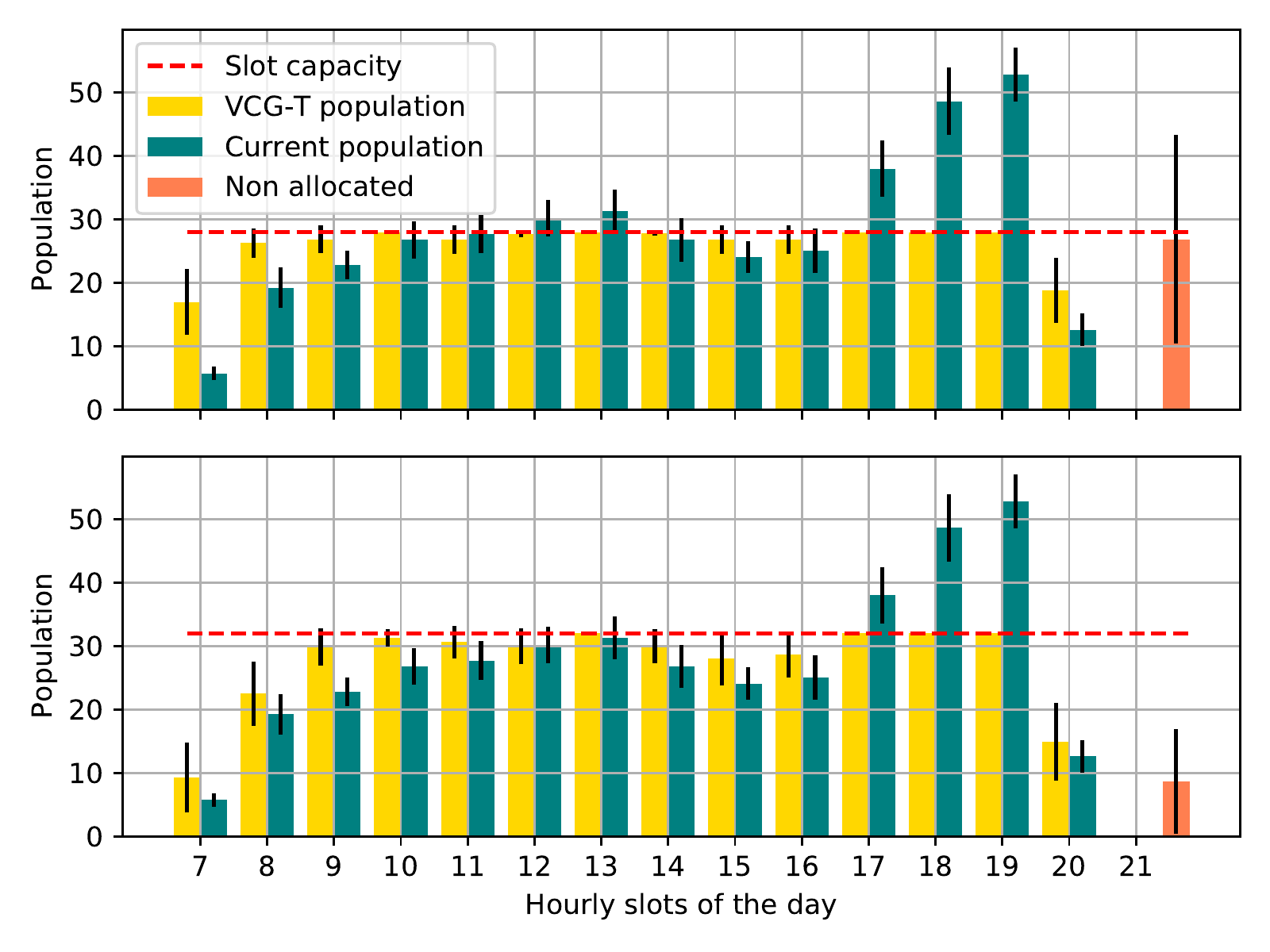}
    \caption{Social congestion reduction, slot capacities $28$ (top) and $32$ (bottom).}
    \label{fig:congestion}
    \centering
    \includegraphics[width=\linewidth]{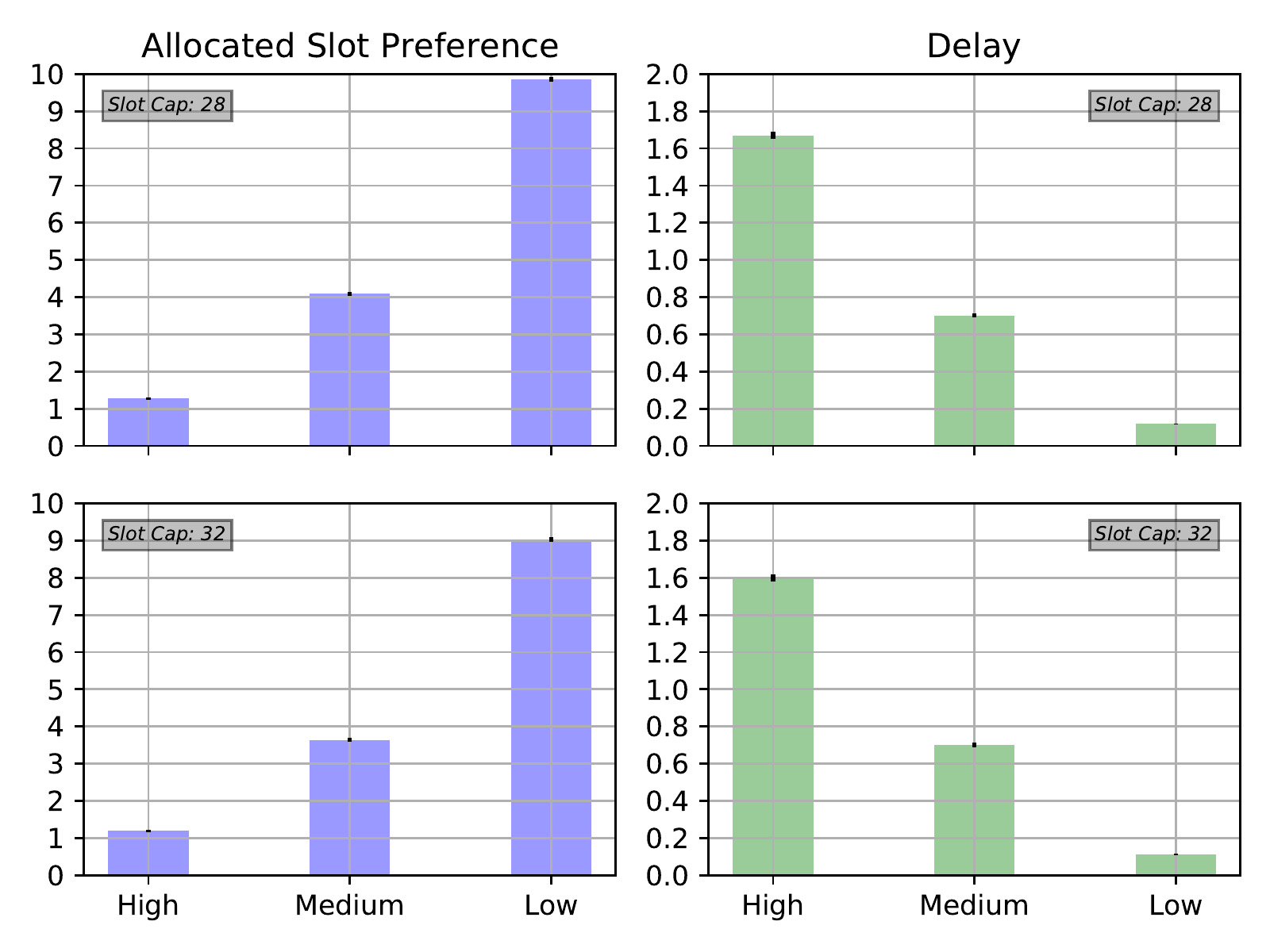}
    \caption{Priority and delay trade-off of \mech.}
    \label{fig:trade-off}
\end{multicols}
\end{figure}

We divide the store opening hours into 14 hourly slots between 7 AM to 9 PM. The valuations of every customer is drawn from a distribution \{\hi:0.1, \med:0.3, \lo:0.6\}.
In this experiment, an agent who is not allocated any slot under \mech\ on a day is removed from the system and counted separately. 
% The experiment is repeated {\bf x} times. % fill value
% The objective is to observe the daily unallocated population.
% Under \mech, the schedules are decided for a day. If a user is not allocated a slot on a certain day, she is given an option to update her importance and preferences for the next day. For the experiments, we assume that the user increases the importance by one level, e.g., \lo\ becomes \med\ and \med\ becomes \hi, and keeps the slot preferences the same. After three consecutive days, if an agent is not allocated, she is considered `non-allocated'. 
\Cref{fig:congestion} shows the comparison of the average current population with that allocated by \mech\ for slot capacities of 28 (above) and 32 (below). The figures also show the daily non-allocated population in red.  Each plot in this section shows the average values with 95\% confidence interval. The plots show the trade-off between better social distancing (lower slot capacity) and its cost (non-allocation). However, in both these cases, the social congestion is reduced by approximately 50\% during the rush hours.

\mech\ also prioritizes the jobs at a cost. The rows of \Cref{fig:trade-off} show the allocated slot preference and the delays for the three difference classes of valuations for the slot capacities 28 and 32 respectively. It shows that a higher valuation comes with a higher delay.

\subsection{Suboptimality and Complexity Reduction (\MAA)}
\label{sec:approx-plot}

The sub-optimality of \MAA\ (\Cref{alg:alg2}) was obtained for a worst-case scenario in \S\ref{sec:multi}. Here we investigate the sub-optimality of \MAA\ and the amount of time it reduces w.r.t.\ an algorithm that finds the optimal allocation of the slots. The top plot of \Cref{fig:tradeoff} shows the percentage reduction ($(t_\text{\textbf{\texttt{OPT}}} - t_\text{\MAA})/t_\text{\textbf{\texttt{OPT}}}$) in the running time of \MAA\ compared to the optimal mechanism,
% (which does an exhaustive search as no better algorithm is known for that class of problems), 
where $t_\text{\textbf{\texttt{OPT}}}$ and $t_\text{\MAA}$ are the running times of the optimal and \MAA\ mechanisms respectively. 

The bottom plot shows the ratio of the optimal social welfare to the welfare yielded by \MAA. The experiment is run with $n=6, k=5$, and $m$ varying from $3$ to $8$. For each value of $m,n,k$, the experiment is repeated $100$ times. The parameters are chosen such that the optimal mechanism is computable in a reasonable time, yet the experiment yields an insightful result. We see that \MAA\ reduces the running time by more than $99.5\%$ and yields an approximation of roughly $1.75$ on an average. 
   
\begin{figure}[h]
\begin{multicols}{2}
\centering
    \includegraphics[width=\linewidth]{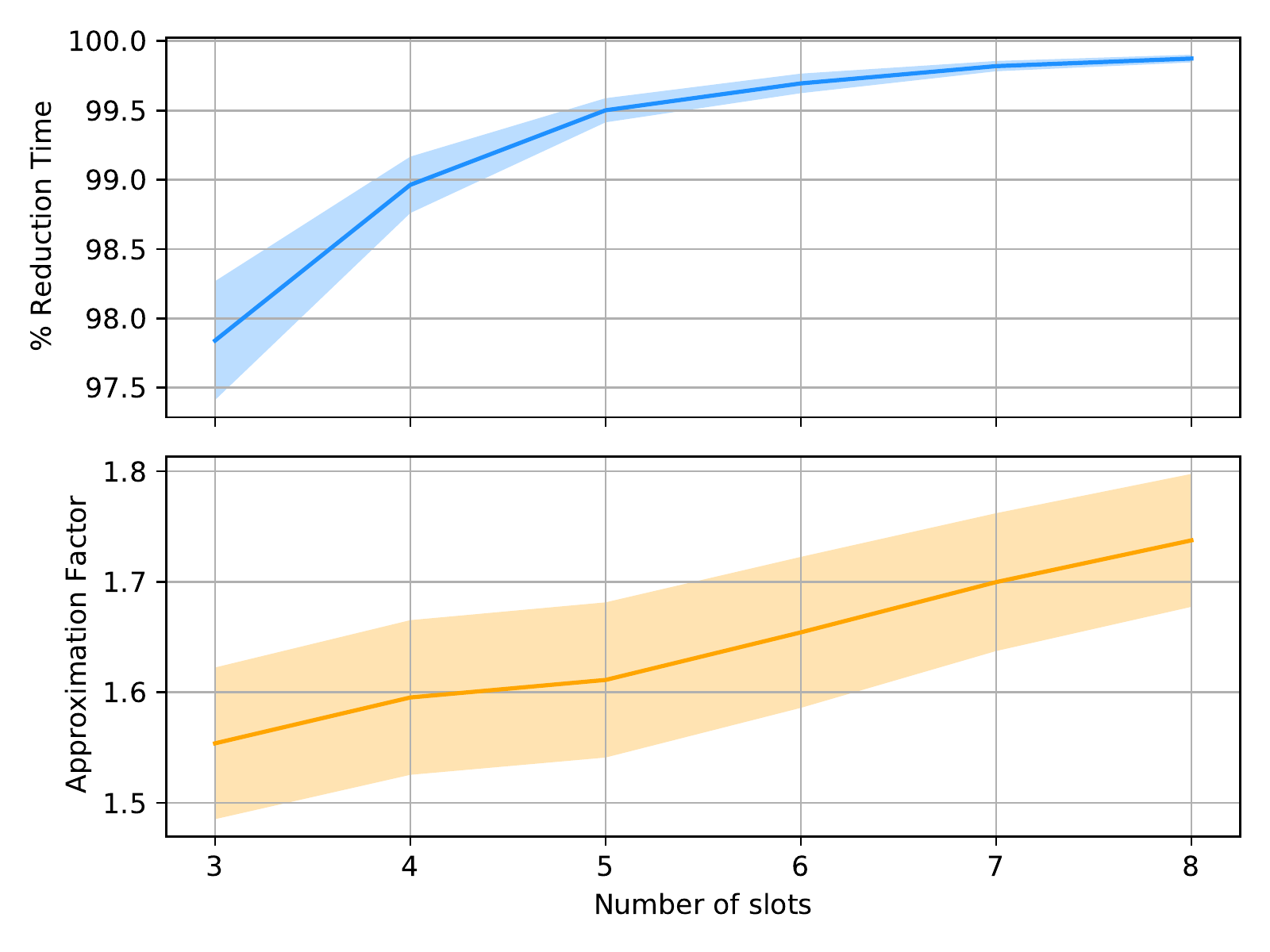}
    \caption{Running time and approximation factor trade-off for \MAA.}
    \label{fig:tradeoff}
    \centering
    \includegraphics[width=\linewidth]{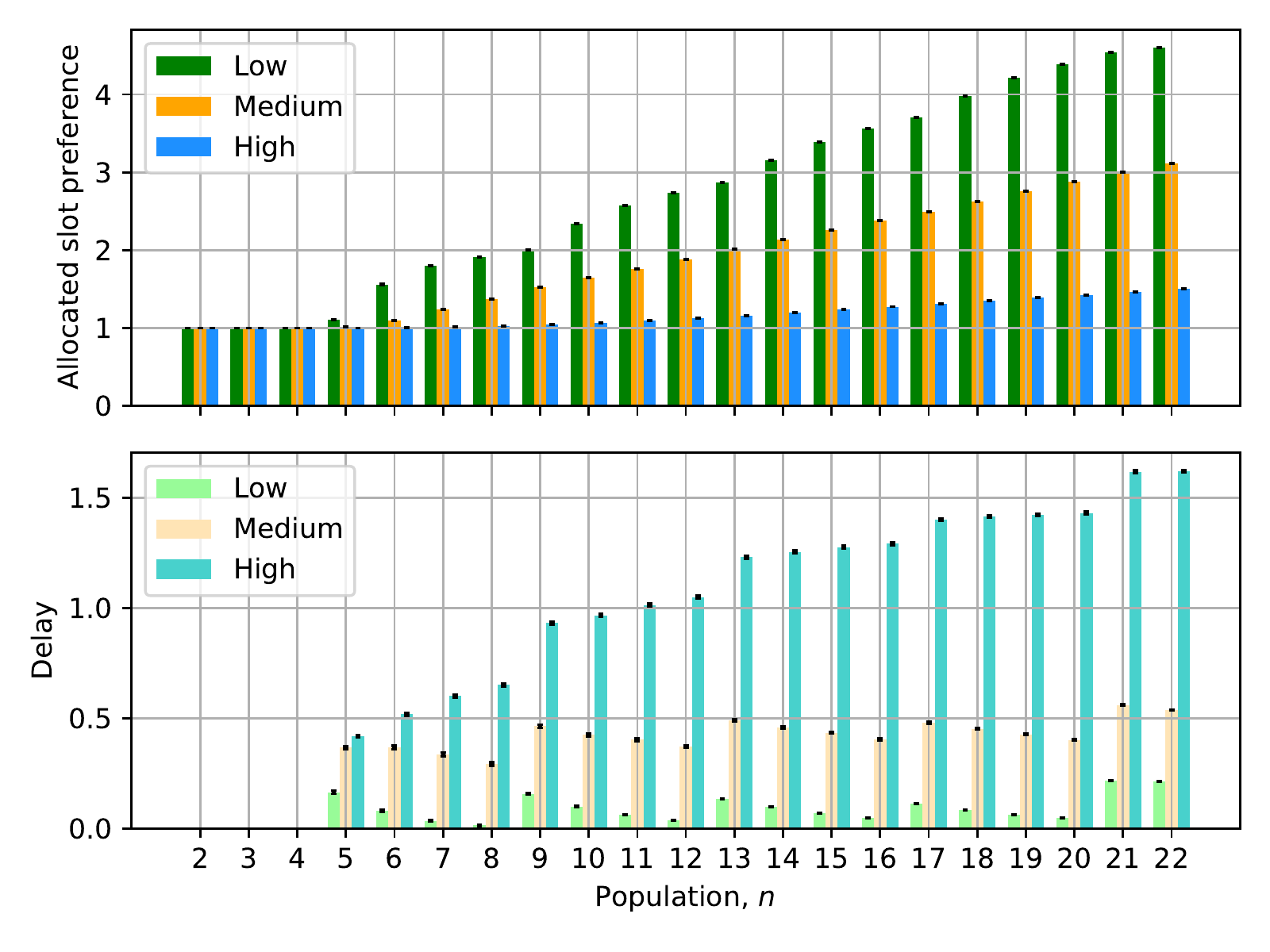}
    \caption{Priority-delay trade-off for \mech.}
    \label{fig:priority-delay}
\end{multicols}
\end{figure}
\subsection{Prioritizing profile and its cost}
\label{sec:priority}

In this section, we investigate what the typical priority slots allotted to an agent of a specific class in \mech\ are. The top plot of \Cref{fig:priority-delay} shows the agents' allocated slot preferences (mean with one standard deviation) versus the population ($n$) plot where $m=5, k=4$, and $\delta=0.65$. The importance of an agent is picked uniformly at random. Values of $n$ vary between $2$ to $1.1 mk$ in steps of 1 (for a population beyond $mk$, some agents have to be dropped). The experiment is repeated $100$ times for every $n$. The plot shows that the higher the importance, the lower is the allocated slot preference for the agents, which is desirable.
However, \mech\ does the prioritized allocations of the agents at the cost of their delays. The bottom plot of \Cref{fig:priority-delay} shows the corresponding delays decided by \mech\ for each of these three classes. The plot shows that an early slot allocation of an agent because of her importance also comes with a longer delay and shows the trade-off between these two decisions.
% When the slot preferences of the agents are picked uniformly at random, the allocated slot preferences of different classes are close to each other, and the payments are also significantly low (Fig. 4).
% \begin{figure}[h!]
%     \centering
%     \includegraphics[scale=0.7]{images/time_pop_5_4_10_0.65}
%     \caption{Priority-delay trade-off}
%     \label{fig:priority-delay}
% \end{figure}

\subsection{Scalability}
\label{sec:scalability}

 This section examines \mech's computation time for finding the allocation and delays for a realistic population. We run \mech\ in \textsf{Python} for different number of slots ($m$) with slot capacity ($k$) being 12. For every $m$, we fixed $n = mk$ and repeated the experiment $10n$ times. \Cref{fig:scalability} shows the growth of the computation time of the mechanism. As a reference, to solve the allocation and delays for the store of \Cref{sec:reduction}, it takes about 100 secs. 
The simulations have been performed in a 64-bit Ubuntu 18.04 LTS machine with Intel(R) Core(TM) i7-7700HQ CPU @2.80GHz quad-core processors and 16 GB RAM.
\begin{figure}
% \begin{multicols}{2}
\centering
    \includegraphics[width=0.55\linewidth]{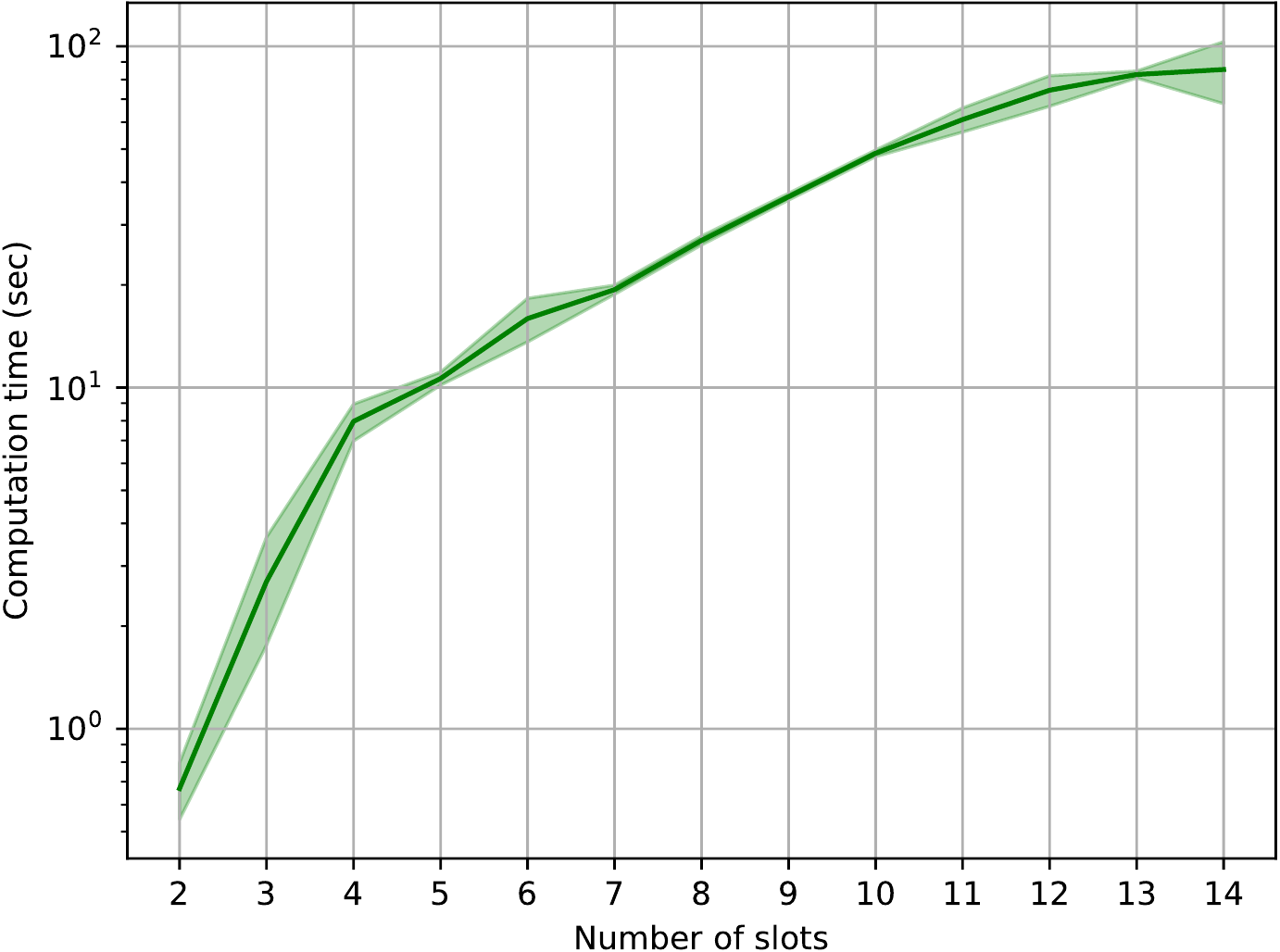}
    \caption{Computation times of \mech.}
    \label{fig:scalability} 
\end{figure}

% \begin{figure}
%     \centering
%     \includegraphics[scale=0.7]{images/scalability}
%     \caption{Computation times of \mech}
%     \label{fig:scalability}
% \end{figure}

% \section{Summary and Plans of Extension}
% \label{sec:summary}

\smallskip \noindent
{\bf Epilogue}: This paper provides a solution to social distancing using social scheduling keeping the COVID-19 pandemic as a motivation, but the solution apply to more general settings too. 
% with capacitated slot allocation problems.
% Experiments show that it prioritizes users based on their importance (at the cost of a cooling-off time). It also reduces the rush-time congestion by a significant proportion without dropping the customers. 
We have already developed an app that runs \mech.
% Handling multiple stores simultaneously is a challenging open problem that we want to address in the future.

% \newpage
%
% ---- Bibliography ----
%
% BibTeX users should specify bibliography style 'splncs04'.
% References will then be sorted and formatted in the correct style.
%
\clearpage
\bibliographystyle{abbrvnat}
% \bibliography{plannait}
\bibliography{abb,swaprava,ultimate,references,references_robots,master}
\clearpage
% \appendix
% \input{app}

\end{document}